\documentclass[letterpaper, 10 pt, conference]{ieeeconf}  

\IEEEoverridecommandlockouts                              
\usepackage{graphicx} 
\usepackage{epsfig} 
\usepackage{amsmath} 
\usepackage{amssymb}  
\usepackage{empheq}
\usepackage{arydshln}
\usepackage{subfig}

\let\labelindent\relax
\usepackage{enumitem}

\usepackage{amsthm}

\usepackage{blindtext}

\title{\LARGE \bf
SIRS Epidemics on Complex Networks: Concurrence of Exact Markov Chain and Approximated Models
}

\author{Navid Azizan Ruhi 
\thanks{Department of Electrical Engineering, University of Southern California, Los Angeles, CA 90089, USA
        {\tt\small azizanru at usc.edu}}%
and Babak Hassibi
\thanks{Department of Electrical Engineering, California Institute of Technology, Pasadena, CA 91125, USA
        {\tt\small hassibi at caltech.edu}}%
}

\begin{document}

\newtheorem{theorem}{Theorem}[section]
\newtheorem{lemma}[theorem]{Lemma}
\newtheorem{proposition}[theorem]{Proposition}
\newtheorem{corollary}[theorem]{Corollary}

\maketitle
\thispagestyle{empty}
\pagestyle{empty}

\begin{abstract}

We study the SIRS (Susceptible-Infected-Recovered-Susceptible) spreading processes over complex networks, by considering its exact $3^n$-state Markov chain model. The Markov chain model exhibits an interesting connection with its $2n$-state nonlinear ``mean-field'' approximation and the latter's corresponding linear approximation. We show that under the specific threshold where the disease-free state is a globally stable fixed point of both the linear and nonlinear models, the exact underlying Markov chain has an $O(\log n)$ mixing time, which means the epidemic dies out quickly. In fact, the epidemic eradication condition coincides for all the three models. Furthermore, when the threshold condition is violated, which indicates that the linear model is not stable, we show that there exists a unique second fixed point for the nonlinear model, which corresponds to the endemic state. We also investigate the effect of adding immunization to the SIRS epidemics by introducing two different models, depending on the efficacy of the vaccine. Our results indicate that immunization improves the threshold of epidemic eradication. Furthermore, the common threshold for fast-mixing of the Markov chain and global stability of the disease-free fixed point improves by the same factor for the vaccination-dominant model.
\end{abstract}

\section{Introduction}

Epidemic models have been extensively studied since a first mathematical formulation was introduced in 1927 by Kermack and McKendrick \cite{kermack1927contribution}. Though initially proposed to understand the spread of contagious diseases \cite{bailey1975mathematical}, the study of epidemics applies to many other areas, such as network security \cite{alpcan2010network,acemoglu2013network}, viral advertising \cite{phelps2004viral,richardson2002mining}, and information propagation \cite{jacquet2010information,cha2009measurement}. Questions of interest include the existence of fixed-points, stability (does the epidemic die out), transient behavior, the cost of an epidemic, how best to control an epidemic, etc.

We consider the spread of an epidemic over a network using the SIRS (susceptible-infected-recovered-susceptible) model where each node can be in one of three states (e.g. as in \cite{zhang2008global}, but network-based). During each time epoch, nodes in the susceptible state can be infected by their infected neighbors according to independent events with probability $\beta$ (the {\em infection rate}) each. Nodes that are infected, during each such time epoch can recover with probability $\delta$ (the {\em recovery rate}) and, finally, nodes in the recovered state can randomly transition to the susceptible state with probability $\gamma$. We will also consider a model which allows for random vaccinations (with probability $\theta$) that permits direct transition from the susceptible state to the recovered one. In its entirety, for a network with $n$ nodes, this yields a Markov chain with $3^n$ states. Ostensibly, because analyzing this Markov chain is too complicated,
various $2n$-dimensional linear and non-linear approximations have been proposed. The most common of these are the $2n$-dimensional non-linear mean-field approximation, and its corresponding linearization about the disease-free fixed point.

Our paper generalizes the analysis of \cite{ahn2013global,ahn2014mixing}, that was concerned with SIS (susceptible-infected-susceptible) models, to the more realistic SIRS case. As in \cite{ahn2013global,ahn2014mixing}, we provide a complete global analysis of the epidemic dynamics for the nonlinear mean-field model. In particular, we show that depending on the largest eigenvalue of the underlying graph adjacency matrix and the ratio of the infection and recovery rates, the global dynamics takes on one of two forms: either the epidemic does out, or it converges to another unique fixed point (the so-called endemic state where a constant fraction of the nodes remain infected). Finally, we tie in these results with the ``true'' underlying Markov chain model and show that the global stability of the $2n$-dimensional approximate models is related to whether the Markov chain is ``fast-mixing'' or not. 

Our paper focuses on discrete-time models (Markov chains and their low-dimensional discrete-time approximations). Continuous-time-discrete-space models, called continuous-time Markov chains have been studied by Draief, Ganesh et al. \cite{draief2006thresholds} and Mieghem et al. \cite{van2009virus}. Continuous-time mean-field approximations for such models have been studied in \cite{shuai2013global,khanafer2014stability,fall2007epidemiological}, where using techniques from Lyapunov theory and the theory of positive systems (and somewhat different from those used here and in 
\cite{ahn2013global,ahn2014mixing}), global stability results and thresholds are obtained for the disease-free and endemic states. However, contrary to the current paper, none of these make an explicit connection to the mixing time of the underlying Markov process.

\section{Model Description}

\subsection{Exact Markov Chain Model}
We start with the exact Markov chain model. For the connected network $G$ with adjacency matrix $A$, let $n$ represent the number of nodes, and $N_i$ the set of neighbors of node $i$. The state of node $i$ at time $t$, denoted by $\xi_i(t)$, can take one of the following values: $0$ for \emph{Susceptible} (or healthy), $1$ for \emph{Infected} (or Infectious), and $2$ for \emph{Recovered}. i.e. $\xi_i(t) \in \left\{0,1,2\right\}$. Fig. \ref{fig} shows the three states and the corresponding transitions. $\beta$ is the transmission probability on each link, $\delta$ is the healing probability, and $\gamma$ is the immunization loss probability. The state of the whole network can be represented as:
\begin{equation}
\xi(t)=(\xi_i(t),\dots,\xi_n(t)) \in \left\{0,1,2\right\}^n
\end{equation}
\begin{figure}[tpb]
  \centering
    \includegraphics[width=0.7\columnwidth]{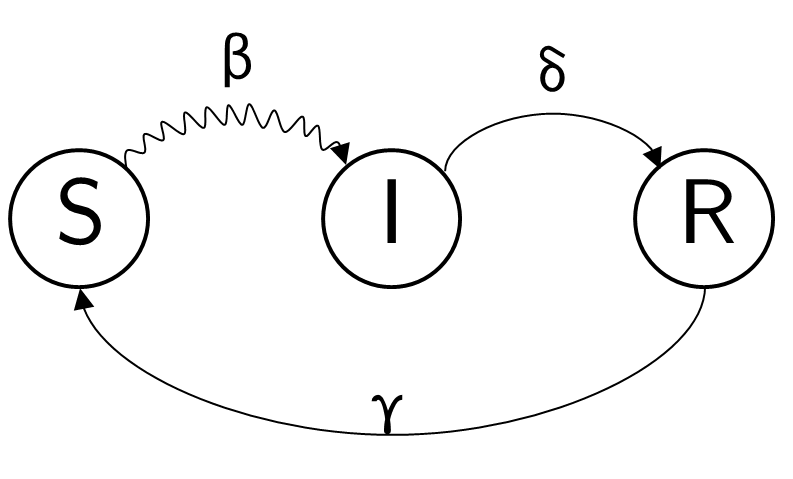}
      \caption{State diagram of a single node in SIRS model, and the transition rates. Wavy arrow represents exogenous (network-based) transition. $S$ is healthy but can get infected, $I$ is infected, $R$ is healthy but cannot get infected.}
      \label{fig}
\end{figure}

Furthermore, let $S$ denote the $3^n\times 3^n$ state transition matrix of the Markov chain, with elements of the form:
\begin{align}\label{MC1}
S_{X,Y}&= \mathbb{P}\left\{\xi(t+1)=Y \mid \xi(t) = X\right\}\notag\\
 &= \prod_{i=1}^n \mathbb{P}\left\{\xi_i(t+1)=Y_i \mid \xi(t) = X\right\} ,
\end{align}
due to the independence of the next states given the current state.
\begin{multline}\label{MC2}
\mathbb{P}\left\{\xi_i(t+1)=Y_i \mid \xi(t) = X\right\}=\\
\begin{cases}
(1-\beta)^{m_i},& \text{if } (X_i,Y_i)=(0,0)\\
1-(1-\beta)^{m_i},& \text{if } (X_i,Y_i)=(0,1)\\
0,& \text{if } (X_i,Y_i)=(0,2)\\
0,& \text{if } (X_i,Y_i)=(1,0)\\
1-\delta,& \text{if } (X_i,Y_i)=(1,1)\\
\delta,& \text{if } (X_i,Y_i)=(1,2)\\
\gamma,& \text{if } (X_i,Y_i)=(2,0)\\
0,& \text{if } (X_i,Y_i)=(2,1)\\
1-\gamma,& \text{if } (X_i,Y_i)=(2,2)\\
\end{cases} ,
\end{multline}
where $m_i=\left\vert{\left\{ {j\in N_i} \mid X_j=1\right\}}\right\vert=\left\vert{N_i\cap I(t)}\right\vert$. The set of susceptible, infected, and recovered nodes at time $t$ are denoted as $S(t)$, $I(t)$, and $R(t)$ respectively.

We state the marginal probability of the nodes as $p_{R,i}(t)$ and $p_{I,i}(t)$, for the probability that \emph{node $i$ is in state $R$ at time $t$} and the probability that \emph{node $i$ is in state $I$ at time $t$}, respectively. Then $p_{S,i}(t)$ follows immediately as $1-p_{R,i}(t)-p_{I,i}(t)$. Based on the abovementioned transition rates, we can calculate these marginal probabilities as:
\begin{equation}
\begin{split}
p_{S,i}(t+1) =&\mathbb{P}\left\{i\in S(t+1) \mid i\in S(t)\right\}p_{S,i}(t)\\
&+ \mathbb{P}\left\{i\in S(t+1) \mid i\in R(t)\right\}p_{R,i}(t) ,
\end{split}
\end{equation}
\begin{equation}
\begin{split}
p_{I,i}(t+1) =&\mathbb{P}\left\{i\in I(t+1) \mid i\in I(t)\right\}p_{I,i}(t)\\
&+ \mathbb{P}\left\{i\in I(t+1) \mid i\in S(t)\right\}p_{S,i}(t) ,
\end{split}
\end{equation}
\begin{equation}
\begin{split}
p_{R,i}(t+1) =&\mathbb{P}\left\{i\in R(t+1) \mid i\in R(t)\right\}p_{R,i}(t)\\
&+ \mathbb{P}\left\{i\in R(t+1) \mid i\in I(t)\right\}p_{I,i}(t) ,
\end{split}
\end{equation}
which yeids:
\begin{empheq}[box=\fbox]{align}
p_{R,i}&(t+1) =(1-\gamma)p_{R,i}(t)+\delta p_{I,i}(t) ,\label{exact_R}\\
p_{I,i}&(t+1) =(1-\delta)p_{I,i}(t)\notag\\
&+(1-(1-\beta)^{m_i})(1-p_{R,i}(t)-p_{I,i}(t)) ,\label{exact_I}
\end{empheq}
and
\begin{equation}\label{exact_S}
p_{S,i}(t+1) =(1-\beta)^{m_i} (1-p_{R,i}(t)-p_{I,i}(t))+\gamma p_{R,i}(t) .
\end{equation}
This is consistent with the fact that $p_{S,i}(t)+p_{I,i}(t)+p_{R,i}(t)=1$ for all $t$.

\subsection{Nonlinear Model}
One may consider the mean-field approximation of the above marginal probabilities, which can be expressed as:
\begin{empheq}[box=\fbox]{align}
P&_{R,i}(t+1) =(1-\gamma)P_{R,i}(t)+\delta P_{I,i}(t) ,\label{nonlinear_R}\\
P&_{I,i}(t+1) =(1-\delta)P_{I,i}(t)+\notag\\
&(1-\prod_{j\in N_i} (1-\beta P_{I,j}(t)))(1-P_{R,i}(t)-P_{I,i}(t)) ,\label{nonlinear_I}
\end{empheq}
and
\begin{multline}\label{nonlinear_S}
P_{S,i}(t+1)=\prod_{j\in N_i} (1-\beta P_{I,j}(t)) (1-P_{R,i}(t)-P_{I,i}(t))\\
+\gamma P_{R,i}(t) .
\end{multline}
We use capital $P$ for the approximated probabilities, to distinguish them from the exact probabilities of the Markov chain, $p$. This approximate model is in fact a nonlinear mapping with $2n$ states (rather than $3^n$ states).

\subsection{Linear Model}
One step further would be to approximate the preceding equations by a linear model. Linearizing Eqs. \eqref{nonlinear_R} and \eqref{nonlinear_I} around the origin results in the following mapping:
\begin{align}
\tilde{P}_{R,i}(t+1) &=(1-\gamma)\tilde{P}_{R,i}(t)+\delta \tilde{P}_{I,i}(t) ,\\
\tilde{P}_{I,i}(t+1) &=(1-\delta)\tilde{P}_{I,i}(t)+\beta\sum\limits_{j\in N_i} \tilde{P}_{I,j} .
\end{align}
These equations (for all $i$) can be expressed in a matrix form:
\begin{empheq}[box=\fbox]{align}\label{linear}
&\begin{bmatrix}\tilde{P}_R(t+1)\\\tilde{P}_I(t+1)\end{bmatrix}=M \begin{bmatrix}\tilde{P}_R(t)\\\tilde{P}_I(t)\end{bmatrix} ,\\
&\text{where}\notag\\
&M = \begin{bmatrix}
(1-\gamma)I_n & \delta I_n\\
0_{n\times n} & (1-\delta)I_n+\beta A
\end{bmatrix} .
\end{empheq}

\section{Epidemic Eradication ($\frac{\beta\lambda_{\max}(A)}{\delta}<1$)}

\subsection{The Trivial Fixed Point of the Map and Steady State of the MC}
The origin is trivially a fixed point of both the linear (Eq. \ref{linear}) and nonlinear (Eqs. \ref{nonlinear_R} and \ref{nonlinear_I}) mappings. In fact, at this fixed point we have:
$$[P_{R,1}(t),  \dots, P_{R,n}(t), P_{I,1}(t), \dots, P_{I,n}(t)]^T= 0_{2n} ,$$
which means all the nodes are susceptible (healthy) with probability 1, and the system stays there permanently, because there are no infected nodes anymore.

More importantly, since the graph $G$ is connected, the Markov chain is irreducible, and it can be seen that its (unique) stationary distribution is:
$$\pi=e_{\bar{0}} ,$$
where $e_X \in \mathbb{R}^{3^n}$ denotes the probability vector with all elements of zero, except the $X$-th one. This means that the steady state of the Markov chain model is $\xi=(0,0,\dots,0) = \bar{0}$, which coincides with the abovementioned fixed point of the mappings. However, the main concern is whether the Markov chain converges to its stationary distribution within a ``reasonable amount of time,'' or not.

\subsection{Stability of the Trivial Fixed Point}
Clearly, if $\|M\|<1$, then the origin is globally stable for the linear model (\ref{linear}) and also locally stable for the nonlinear model (\ref{nonlinear_I}, \ref{nonlinear_R}). The eigenvalues of $M$ matrix consist of the eigenvalues of $(1-\gamma)I_n$ and the eigenvalues of $(1-\delta)I_n+\beta A$. Noticing that the eigenvalues of $(1-\gamma)I_n$ are always less than one, it can be concluded that $\|M\|<1$ if the largest eigenvalue of $(1-\delta)I_n+\beta A$ is less than one.

In addition, the linear model (\ref{linear}) is an upperbound on the nonlinear model (\ref{nonlinear_R}, \ref{nonlinear_I}), i.e.
\begin{multline}
P_{I,i}(t+1) =(1-\delta)P_{I,i}(t)\\
+(1-\prod_{j\in N_i} (1-\beta P_{I,j}(t)))(1-P_{R,i}(t)-P_{I,i}(t))\\
\leq (1-\delta)P_{I,i}(t)+\beta\sum\limits_{j\in N_i} P_{I,j} ,
\end{multline}
and consequently
\begin{equation}
P_I(t+1)\preceq((1-\delta)I_n+\beta A)P_I(t) ,
\end{equation}
where $\preceq$ denotes that the inequality holds element-wise for all the elements. This concludes the following result.
\begin{proposition}
If $\frac{\beta\lambda_{\max}(A)}{\delta}<1$, then the origin is a globally stable fixed point for both linear model (\ref{linear}) and nonlinear model (\ref{nonlinear_R}, \ref{nonlinear_I}).
\end{proposition}

\subsection{Mixing Time of the MC}
We further show that when $\frac{\beta\lambda_{\max}(A)}{\delta}<1$, not only are the linear and nonlinear maps globally stable at the origin, but also the mixing time of the Markov chain is $O(\log n )$, meaning that the Markov chain mixes fast and the epidemic dies out. This result has been shown for the simpler case of SIS model in \cite{ahn2014mixing}.

Let the row vector $\mu(t)\in \mathbb{R}^{3^n}$ be the probability vector of the Markov chain. The relationship between these probabilities ($\mu_X(t)$) and the marginal probabilities ($p_{R,i}(t)$, $p_{I,i}(t)$) is in the following forms: $p_{R,i}(t)=\sum_{X_i=2} \mu_X(t)$, $p_{I,i}(t)=\sum_{X_i=1} \mu_X(t)$. We express all these terms as well as $p_0=\sum \mu_X(t)=1$ in the form of a column vector $p(t)=[p_0(t),p_1(t),\dots,p_{2n}]^T$. i.e.
\begin{equation}
p(t)=\begin{bmatrix} 1, \vline p_{R,1}(t), \dots,p_{R,n}(t), \vline p_{I,1}(t), \dots, p_{I,n}(t)\end{bmatrix}^T .
\end{equation}
The matrix $B \in \mathbb{R}^{3^n\times (2n+1)}$ which relates the \emph{observable data}, $p(t)$, and the hidden \emph{complete data}, $\mu(t)$, can be expressed as:
\begin{equation}
B_{X,k}=
\begin{cases}
1,& \text{if } k=0\\ \hdashline
0,& \text{if } k\in\left\{1,2,\dots,n\right\} \text{ and } X_k=0\\
0,& \text{if } k\in\left\{1,2,\dots,n\right\} \text{ and } X_k=1\\
1,& \text{if } k\in\left\{1,2,\dots,n\right\} \text{ and } X_k=2\\ \hdashline
0,& \text{if } k\in\left\{n+1,n+2,\dots,2n\right\} \text{ and } X_{k-n}=0\\
1,& \text{if } k\in\left\{n+1,n+2,\dots,2n\right\} \text{ and } X_{k-n}=1\\
0,& \text{if } k\in\left\{n+1,n+2,\dots,2n\right\} \text{ and } X_{k-n}=2\\
\end{cases}
\end{equation}

Now we are ready to proceed to the main theorem of this section.
\begin{theorem}\label{thm_mixing}
If $\frac{\beta\lambda_{\max}(A)}{\delta}<1$, the mixing time of the Markov chain whose transition matrix $S$ is described by Eqs. \eqref{MC1} and \eqref{MC2} is $O(\log n)$.
\end{theorem}
\begin{proof}
First we use a linear programming technique to show for each $i\in\left\{n+1,n+2,\dots,2n\right\}$, we have $p_i(t+1) \leq (1-\delta)p_i(t) + \beta\sum\limits_{j\in N_i}p_j(t)$.
Let $f_i\in \mathbb{R}^{2n+1}$ represent the $i^{th}$ unit column vector.
For the sake of convenience, let us drop the time index $(t)$.
\begin{align}
\max_{\mu B=p^T, \mu\succeq 0} p_i(t+1) &= \max_{\mu B=p^T, \mu\succeq 0} \mu SBf_i\\
 &= \max_{\mu\succeq 0} \min_\lambda \mu SBf_i-(\mu B-p^T)\lambda\\ \label{minmax}
 &= \min_\lambda \max_{\mu\succeq 0} \mu(SBf_i-B\lambda) +p^T\lambda ,
\end{align}
where $\lambda\in\mathbb{R}^{2n+1}$ is a column vector. If any element of $(SBf_i-B\lambda)$ is strictly positive, it leads to $\max_{\mu\succeq 0} \mu(SBf_i-B\lambda)=+\infty$. Therefore:
\begin{equation}\label{leq}
SBf_i-B\lambda\preceq 0 .
\end{equation}
Now we proceed with further calculation of $SBf_i$ and $B\lambda$.
\begin{align}
&(SBf_i)_X = (SB)_{X,i} = \sum_{Y\in\left\{0,1,2\right\}^n} S_{X,Y}B_{Y,i}\\
 &= \begin{cases} \mathbb{P}\left\{ Y_i=2 \mid X\right\},& i\in\left\{1,2,\dots,n\right\}\\ \mathbb{P}\left\{ Y_{i-n}=1 \mid X\right\},& i\in\left\{n+1,n+2,\dots,2n\right\}\end{cases}\\
 &= \begin{cases}
0,& \text{if } i\in\left\{1,2,\dots,n\right\} \text{ and } X_i=0\\
\delta,& \text{if } i\in\left\{1,2,\dots,n\right\} \text{ and } X_i=1\\
1-\gamma,& \text{if } i\in\left\{1,2,\dots,n\right\} \text{ and } X_i=2\\
1-(1-\beta)^{m_{i-n}},& \text{if } i\in\left\{n+1,\dots,2n\right\} \text{ and } X_{i-n}=0\\
1-\delta,& \text{if } i\in\left\{n+1,\dots,2n\right\} \text{ and } X_{i-n}=1\\ \label{SBf_i}
0,& \text{if } i\in\left\{n+1,\dots,2n\right\} \text{ and } X_{i-n}=2\\
\end{cases}\\ \label{Bl}
&(B\lambda)_X =\lambda_0 + \sum\limits_{k=1}^n B_{X,k}\lambda_k + \sum\limits_{k=n+1}^{2n} B_{X,k}\lambda_k .
\end{align}
As mentioned earlier, we want to evaluate $p_i(t+1)$ only for $i\in\left\{n+1,n+2,\dots,2n\right\}$.
Define $\hat{i} \in \left\{0,1,2\right\}^n$ as the state where only $i$ is infected, and the rest are susceptible. Trying several $X$ in \eqref{leq} using \eqref{SBf_i} and \eqref{Bl} yields:
\begin{equation}
\begin{cases}
X=\bar{0}, &\lambda_0+0+0 \geq 0\\
X=\bar{2}, &\lambda_0+\sum\limits_{k=1}^n \lambda_k+0 \geq 0\\
X=\hat{i}, &\lambda_0+0+\lambda_{n+i}\geq 1-\delta\\
X=\hat{j}, j\in N_i, &\lambda_0+0+\lambda_{n+j}\geq \beta\\
X=\hat{j}, j\not\in N_i, &\lambda_0+0+\lambda_{n+j}\geq 0
\end{cases}
\end{equation}
Now we claim that $\lambda^*=[\lambda_0^*,\lambda_1^*,\dots,\lambda_{2n}^*]^T$ defined by the following values is in the feasible set:
\begin{equation}
\begin{cases}
\lambda_0^*=0\\
\lambda_1^*=\dots=\lambda_n^*=0\\
\lambda_{n+i}^*=1-\delta\\
\lambda_{n+j}=\beta \text{  for } j\in N_i\\
\lambda_{n+j}=0 \text{  for } j\not\in N_i
\end{cases}
\end{equation}
We verify the claim for all possible cases as the following.\\
For $X_i=0, \left\vert{N_i\cap I(t)}\right\vert=m$ :
\begin{multline}
\mathbb{P}\left\{ Y_i=1 \mid X\right\}=1-(1-\beta)^m \leq\\
m\beta = \lambda_0^* + \sum\limits_{k=1}^n B_{X,k}\lambda_k^* + \sum\limits_{k=n+1}^{2n} B_{X,k}\lambda_k^* .
\end{multline}
For $X_i=1, \left\vert{N_i\cap I(t)}\right\vert=m$ :
\begin{multline}
\mathbb{P}\left\{ Y_i=1 \mid X\right\}=1-\delta \leq\\
1-\delta+m\beta = \lambda_0^* + \sum\limits_{k=1}^n B_{X,k}\lambda_k^* + \sum\limits_{k=n+1}^{2n} B_{X,k}\lambda_k^* .
\end{multline}
For $X_i=2, \left\vert{N_i\cap I(t)}\right\vert=m$ :
\begin{multline}
\mathbb{P}\left\{ Y_i=1 \mid X\right\}=0 \leq\\
m\beta = \lambda_0^* + \sum\limits_{k=1}^n B_{X,k}\lambda_k^* + \sum\limits_{k=n+1}^{2n} B_{X,k}\lambda_k^* .
\end{multline}
It follows that $\lambda^*$ is in the feasible set. Back to the Eq. \eqref{minmax} we have:
\begin{align}
\max_{\mu B=p^T, \mu\succeq 0} p_i(t+1) &= \min_\lambda \max_{\mu\succeq 0} \mu(SBf_i-B\lambda) +p^T\lambda\notag\\
 &\leq p^T\lambda^* = (1-\delta)p_i +\beta\sum\limits_{j\in N_i}p_j ,
\end{align}
which proves:
\begin{equation}
p_{I,i}(t+1) \leq (1-\delta)p_{I,i}(t) + \beta\sum\limits_{j\in N_i}p_{I,j}(t) .
\end{equation}
Moreover, we already know that $p_{R,i}(t+1)=(1-\gamma)p_{R,i}(t)+\delta p_{I,i}(t)$ (Eq. \ref{exact_R}), and all the equations can be expressed in a vector form, using $p_R=[p_{R,1}(t),p_{R,2}(t),\dots,p_{R,n}(t)]^T$ and $p_I=[p_{I,1}(t),p_{I,2}(t),\dots,p_{I,n}(t)]^T$:
\begin{align}
\begin{bmatrix}p_R\\p_I\end{bmatrix}(t+1)&\preceq\begin{bmatrix}
(1-\gamma)I_n & \delta I_n\\
0_n & (1-\delta)I_n+\beta A
\end{bmatrix} \begin{bmatrix}p_R\\p_I\end{bmatrix}(t)\\
&= M \begin{bmatrix}p_R\\p_I\end{bmatrix}(t). \notag
\end{align}

The definition of the mixing time \cite{levin2009markov} is:
\begin{equation}
t_{mix}(\epsilon)=\min\left\{t : \sup_\mu \|\mu S^t-\pi\|_{TV} \leq \epsilon\right\} .
\end{equation}
More specifically we have:
\begin{align}
\|\mu S^t-\pi\|_{TV} &= \frac{1}{2}\sum\limits_X \lvert(\mu S^t)_X - \pi_X\rvert\\
&= \frac{1}{2}\sum\limits_X \lvert(\mu S^t)_X - (e_{\bar{0}})_X\rvert\\
&= \frac{1}{2}\big( 1-(\mu S^t)_{\bar{0}}\big) + \frac{1}{2}\sum\limits_{X\neq \bar{0}} (\mu S^t)_X\\
&= \frac{1}{2}\big( 1-(\mu S^t)_{\bar{0}}\big) + \frac{1}{2}\big( 1-(\mu S^t)_{\bar{0}}\big)\\
&= 1-(\mu S^t)_{\bar{0}}\\
&= 1-\mu S^t e_{\bar{0}}^T\\
&\leq 1-e_{\bar{1}} S^t e_{\bar{0}}^T  .
\end{align}
Hence, for any $t< t_{mix}(\epsilon)$:
\begin{align}
\epsilon &< 1-\mathbb{P}\left\{ \substack{\text{all nodes are susceptible at time $t$} \mid \\\text{all nodes were infected at time $0$}}\right\}\\
 &=\mathbb{P}\left\{ \substack{\text{some nodes are infected or recovered at time $t$} \mid \\\text{all nodes were infected at time $0$}} \right\}\\
 &\leq \sum\limits_{i=1}^n (p_{I,i}(t)+p_{R,i}(t)) = 1_{2n}^T \begin{bmatrix}p_R\\p_I\end{bmatrix}(t)\\
 &\leq 1_{2n}^T M^t \begin{bmatrix}p_R\\p_I\end{bmatrix}(0)\\
 &\leq 1_{2n}^T M^t 1_{2n}\\
 &\leq \|1_{2n}\|^2 \|M\|^t\\
 &=2n\|M\|^t .
\end{align}
$\|M\|<1$ leads to the fact that $t< \frac{\log \frac{2n}{\epsilon}}{-\log \|M\|}$ for all $t< t_{mix}(\epsilon)$. Therefore $t_{mix}(\epsilon)\leq \frac{\log \frac{2n}{\epsilon}}{-\log \|M\|}$, which means the mixing time is $O(\log n)$.
\end{proof}

\section{Epidemic Spread ($\frac{\beta\lambda_{\max}(A)}{\delta}>1$)}
\subsection{Existence and Uniqueness of Nontrivial Fixed Point}
The trivial fixed point of the mappings, the origin, is not stable if $(1-\delta)+\beta\lambda_{\max}(A)>1$. Moreover, it is not clear in general whether there exists any other fixed point, or how many fixed points exist if so. However, it has been proved in \cite{ahn2013global} that for SIS model indeed there exists a unique nontrivial fixed point when $(1-\delta)+\beta\lambda_{\max}(A)>1$. In this section we extend this result to the more general case of SIRS model.

By rearranging Eq. \eqref{nonlinear_I}, we can rewrite the system equations as:
\begin{empheq}[left=\empheqlbrace]{align}
P_{R,i}(t+1)=&(1-\gamma)P_{R,i}(t)+\delta P_{I,i}(t)\label{nonlinear_R_new}\\
P_{I,i}(t+1)=&P_{I,i}(t)+(1-P_{R,i}(t)-P_{I,i}(t))\notag\\
&\cdot\big(\Xi_i(P_I(t))-\omega(P_{R,i}(t),P_{I,i}(t))\big) ,\label{nonlinear_I_new}
\end{empheq}
where $\Xi_i \colon [0,1]^n \to [0,1]$ and $\omega \colon [0,1]^2 \to \mathbb{R}^+$ are the following maps associated with network $G$:
\begin{equation}
\Xi_i(P_I(t))=1-\prod_{j\in N_i} (1-\beta P_{I,j}(t)) ,
\end{equation}
\begin{equation}
\omega(P_{R,i}(t),P_{I,i}(t))=\frac{\delta P_{I,i}(t)}{1-P_{R,i}(t)-P_{I,i}(t)} .
\end{equation}

It can be verified that the maps defined above, enjoy the following properties:
\begin{enumerate}[label=(\alph*)]
\item $\Xi_i(0_n)=0$\\ $\frac{\partial \Xi_i(P_I)}{\partial P_{I,j}}\bigg|_{0_n}=\beta A_{i,j}$
\item $\begin{cases}\frac{\partial \Xi_i(P_I)}{\partial P_{I,j}}>0 &\mbox{if } i\in N_j\\ \frac{\partial \Xi_i(P_I)}{\partial P_{I,j}}=0 &\mbox{if } i\not\in N_j\end{cases}$
\item $\frac{\partial^2\Xi_i(P_I)}{\partial P_{I,j} \partial P_{I,k}}\leq 0 \quad \forall i,j,k \in \{1,\dots,n\}$
\item $\omega(0,0)=0$\\ $\frac{\partial \omega(P_{R,i},P_{I,i})}{\partial P_{I,i}}\bigg|_{(0,0)}=\delta$
\item $\frac{\partial \omega(P_{R,i},P_{I,i})}{\partial P_{I,i}}>0 \quad \forall P_{I,i}\in(0,1)$
\item $\frac{\omega(P_{R,i},P_{I,i})}{P_{I,i}}$ is an increasing function of both $P_{R,i}$ and $P_{I,i}$. More specifically: $\frac{\omega(s_1,t_1)}{s_1}<\frac{\omega(s_2,t_2)}{s_2}$ if $s_1<s_2$ and $t_1<t_2$.
\end{enumerate}

The main result of this section is as follows.
\begin{theorem}
If $\frac{\beta\lambda_{\max}(A)}{\delta}>1$, the nonlinear map (\ref{nonlinear_R}, \ref{nonlinear_I}), or equivalently (\ref{nonlinear_R_new}, \ref{nonlinear_I_new}), has a unique nontrivial fixed point.
\end{theorem}
\begin{proof}
Let's define the map $\Psi \colon [0,1]^{2n} \to R^n$ as $\Psi=[\Psi_1,\dots,\Psi_n]^\top$ with
\begin{equation}
\Psi_i(P(t))=\Xi_i(P_I(t))-\omega(P_{R,i}(t),P_{I,i}(t)) .
\end{equation}
Note that zeros of $\Psi$ correspond to fixed points of the nonlinear map (Eq. \ref{nonlinear_I_new}).

Now we define sets $U_i$ and $U$ as follows:
\begin{equation}
U_i=\{ x_I\in[0,1]^n: \Psi_i(\begin{bmatrix}x_R\\x_I\end{bmatrix})\geq 0, 0_n\preceq x_R\preceq 1_n-x_I\} ,
\end{equation}
\begin{equation}
U=\bigcap\limits_{i=1}^n U_i .
\end{equation}
In plain words, $U$ is the set of ``infection situations'' from which the system becomes ``more infected'' or remains there.

From Lemma 3.1 in \cite{ahn2013global}, $\lambda_{\max}((1-\delta)I_n+\beta A)>1$ implies that there exists $v\succ0_n$ such that $(\beta A-\delta I_n)v\succ0_n$. On the other hand $\Psi(0_{2n})=0_n$ and the Jacobian of $\Psi$ at the origin is equal to $\begin{bmatrix}0_{n\times n} & \beta A-\delta I_n\end{bmatrix}_{n\times 2n}$. As a result, there exists a small $\epsilon>0$ such that $\Psi(\begin{bmatrix}\epsilon u\\ \epsilon v\end{bmatrix})=(\beta A-\delta I_n)v\epsilon$, which is $\succ 0_n$, and indicates that $\epsilon v \in U$.

We claim that if $x,y\in U$, then $\max(x,y)\triangleq(\max(x_1,y_1),\dots,\max(x_n,y_n))\in U$.
For all $i\in\{1,\dots,n\}$, $\exists\, a_i\in[0,1-x_i]\mbox{ s.t. } \Xi_i(x)-\omega(a_i,x_i)\geq0$, and $\exists\, b_i\in[0,1-y_i]\mbox{ s.t. } \Xi_i(y)-\omega(b_i,y_i)\geq0$.
\begin{align}
\Psi_i(\begin{bmatrix}c\\ \max(x,y)\end{bmatrix}) &=\Xi_i(\max(x,y))-\omega(c_i,\max(x_i,y_i)).
\end{align}
Without loss of generality assume $\max(x_i,y_i)=x_i$, then if we pick $c_i=a_i$, it follows that:
\begin{align}
\Psi_i(\begin{bmatrix}c\\ \max(x,y)\end{bmatrix}) &=\Xi_i(\max(x,y))-\omega(a_i,x_i)\\
&\geq \Xi_i(x)-\omega(a_i,x_i)\geq 0\label{random_ineq_1} .
\end{align}
Inequality \eqref{random_ineq_1} comes from Property (b). Now $\max(x,y)\in U_i$, and we can use the same argument for all $i$. Hence $\max(x,y)\in U$, and the claim is true.

It follows that there exists a unique maximal point $x^*\in U$ such that $x^*\succeq x$ for all $x\in U$. Moreover, since $\epsilon v\in U$, we can conclude that $x^*\succ 0_n$ (all elements of $x^*$ are positive).

Now we further claim that $\Psi_i(\begin{bmatrix}a\\x^*\end{bmatrix})=0$ for some $0_n\preceq a\preceq 1_n-x^*$ and $\forall i\in \{1,\dots,n\}$. Assume, by the way of contradiction, that $\Psi_i(\begin{bmatrix}a\\x^*\end{bmatrix})\neq0$ for all $0_n\preceq a\preceq 1_n-x^*$, which means $\Psi_i(\begin{bmatrix}a\\x^*\end{bmatrix})>0$. Since $\Psi_i(\begin{bmatrix}a\\x^*\end{bmatrix})=\Xi_i(x^*)-\omega(a_i,x_i^*)>0$ and $\omega(a_i,x_i^*)<\omega(a_i,z_i)$ for any $z_i>x_i^*$ (Property (e)), there exists $z_i>x_i^*$ such that
\begin{equation}\label{eq13}
\Xi_i(x^*)-\omega(a_i,z_i)\geq 0 .
\end{equation}
Now define $z=[z_1,\dots,z_n]^\top$ with $z_j=x_j \, \forall j\neq i$. For every $k\in\{1,\dots,n\}$ we have
$$\Psi_k(\begin{bmatrix}a\\z\end{bmatrix})=\Xi_k(z)-\omega(a_k,z_k)\geq \Xi_k(x^*)-\omega(a_k,z_k)\geq 0 ,$$
for some $0_n\preceq a\preceq 1_n-z$. The first inequality holds by Property (b). The second inequality holds by \eqref{eq13} for $k=i$, and by definition for $k\neq i$. It implies that $z\in U$. Since $z_i>x_i^*$, this contradicts the fact that $x^*$ is the maximal point is $U$. Hence $\Psi_i(\begin{bmatrix}a\\x^*\end{bmatrix})=0$ for some $0_n\preceq a\preceq 1_n-x^*$, and this is true for all $i\in\{1,\dots,n\}$. Thus far we have proved that there exists a nontrivial zero for $\Psi$.

We note that in order for a point $\begin{bmatrix}p_R^*\\p_I^*\end{bmatrix}$ to be a fixed point of the nonlinear map, it should satisfy Eq. \eqref{nonlinear_R_new}, i.e.
\begin{equation}\label{relation}
p_{R,i}^*=(1-\gamma)p_{R,i}^*+\delta p_{I,i}^* \implies p_{R,i}^*=\frac{\delta}{\gamma}p_{I,i}^* .
\end{equation}

For proving the uniqueness of nontrivial zero of $\Psi$, assume by contradiction that in addition to $x^*$, $y^*$ is another nontrivial zero. Therefore $y^*\in U$, and $\Psi(\begin{bmatrix}b\\y^*\end{bmatrix})=0_n$ for some $0_n\preceq b\preceq 1_n-y^*$.

We claim that $y^*$ is all-positive. Let us define $K_0=\{1\leq i\leq n : y_i^*=0\}$ and $K_+=\{1\leq i\leq n : y_i^*>0\}$. $K_0\cup K_+=\{1,\dots,n\}$. Assume that $K_0$ is not empty and $k\in K_0$. Since $G$ is connected, there exists $j\in K_+$ such that $j$ is a neighbor of a $k$.
\begin{equation}
\Psi_k(\begin{bmatrix}b\\y^*\end{bmatrix})=\Xi_k(y^*)-\omega(b_k,y_k^*)=\Xi_k(y^*)>0 .
\end{equation}
The second equality holds by Property (d) and due to $b_k=y_k^*=0$ (from Eq. \ref{relation}). The inequality comes from Property (b) ($k\in N_j$) and $y_j^*>0$. This contradicts $\Psi(\begin{bmatrix}b\\y^*\end{bmatrix})=0_n$, and implies that $K_0=\emptyset$, and therefore every element of $y^*$ is positive.

By Property (c), and from Lemma 2.1 in \cite{ahn2013global}, we know for $s\leq 1$
$$\frac{\Xi_i(u+sv)-\Xi_i(u)}{s}\geq \frac{\Xi_i(u+v)-\Xi_i(u)}{1} .$$
By setting $u=0_n$ and $v=x^*$, and using Property (a), it follows that
\begin{equation}\label{nice}
\frac{\Xi_i(sx^*)}{s}\geq \Xi_i(x^*) .
\end{equation}

For $x^*$ and $y^*$ there exists $\alpha\in(0,1)$ such that $y^*\succeq\alpha x^*$ and $y_j^*=\alpha x_j^*$ for some $j\in\{1,\dots,n\}$.
\begin{align}
\Psi_j(\begin{bmatrix}b\\y^*\end{bmatrix}) &=\Xi_j(y^*)-\omega(b_j,\alpha x_j^*)\\
&\geq \Xi_j(\alpha x^*)-\omega(b_j,\alpha x_j^*)\label{rnd_eq_1}\\
&\geq \alpha\Xi_j(x^*)-\omega(b_j,\alpha x_j^*)\label{rnd_eq_2}\\
&> \alpha\Xi_j(x^*)-\alpha \omega(\frac{b_j}{\alpha},x_j^*)\label{rnd_eq_3}\\
&= \alpha\big(\Xi_j(x^*)- \omega(a_j,x_j^*)\big)=0\label{rnd_eq_4} .
\end{align}
Inequality \eqref{rnd_eq_1} holds by Property (b), \eqref{rnd_eq_2} follows from \eqref{nice}, \eqref{rnd_eq_3} holds by Property (f), and finally \eqref{rnd_eq_4} comes from \eqref{relation}. This contradicts that $\Psi_i(\begin{bmatrix}b\\y^*\end{bmatrix})=0$ for all $i$.

It concludes that $\begin{bmatrix}a\\x^*\end{bmatrix}$ is the unique nontrivial zero of $\Psi$, and hence the unique nontrivial fixed point of the system.
\end{proof}

\subsection{Stability of the Nontrivial Fixed Point}
Since the trivial fixed point was globally stable when $\frac{\beta\lambda_{\max}(A)}{\delta}<1$, the existence of a second unique fixed point at $\frac{\beta\lambda_{\max}(A)}{\delta}>1$ raises the question of whether it is also stable. However, it turns out that this is not true in general. In fact, same as immune-admitting SIS model in \cite{ahn2013global}, we can find simple examples in which the system converges to a cycle rather than the unique second fixed point.

Nevertheless, in the immune-admitting SIS, this fixed point has shown to be stable with high probability for Erd\H{o}s-R\'enyi graphs \cite{ahn2013global}. Furthermore, in a variation of SIS model \cite{chakrabarti2008epidemic}, the second fixed point is indeed globally stable.

\section{Vaccination}

In this section we consider the effect of vaccination by incorporating direct immunization into the model studied in the previous sections. In other words, the transition from $S$ to $R$ is also permitted now (See Fig. \ref{fig2}). This class of processes are sometimes referred to as SIV (Susceptible-Infected-Vaccinated) epidemics, although the term is often used for population-based estimated models. Depending on the value of $\gamma$, this model can represent temporary ($\gamma\neq 0$) or permanent ($\gamma=0$) immunization.
\begin{figure}[thpb]
  \centering
    \includegraphics[width=0.7\columnwidth]{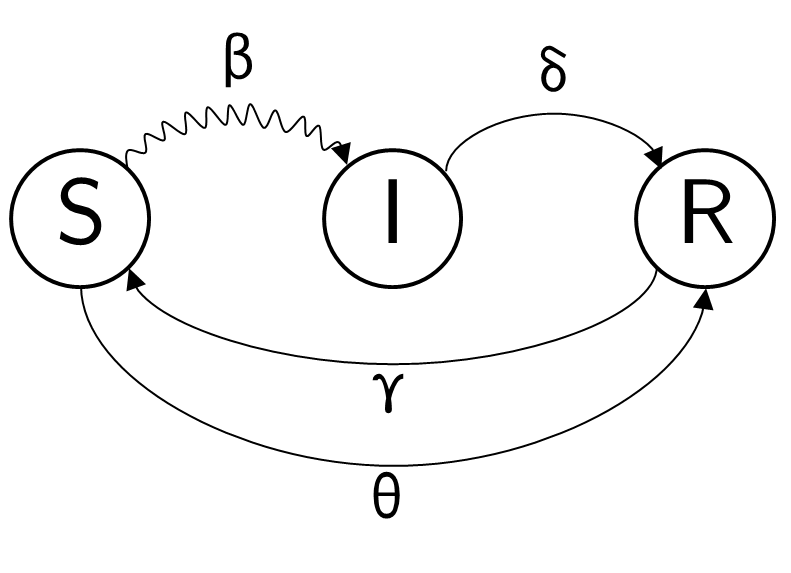}
      \caption{State diagram of a single node in the SIRS-with-Vaccination model, and the transition rates. Wavy arrow represents exogenous (network-based) transition. $\theta$ represents the probability of direct immunization.}
      \label{fig2}
\end{figure}

\subsection{Infection-Dominant Model}

In this case, assuming that the infection is dominant (meaning that if a susceptible node receives both infection and vaccine at the same time, it gets infected), the elements of state transition matrix are
\begin{align}\label{MC1_id}
S_{X,Y}&= \mathbb{P}\left\{\xi(t+1)=Y \mid \xi(t)= X\right\}\notag\\
&= \prod_{i=1}^n \mathbb{P}\left\{\xi_i(t+1)=Y_i \mid \xi(t) = X\right\} ,
\end{align}
where
\begin{multline}\label{MC2_id}
\mathbb{P}\left\{\xi_i(t+1)=Y_i \mid \xi(t) = X\right\}=\\
\begin{cases}
(1-\beta)^{m_i}(1-\theta),& \text{if } (X_i,Y_i)=(0,0)\\
1-(1-\beta)^{m_i},& \text{if } (X_i,Y_i)=(0,1)\\
(1-\beta)^{m_i}\theta,& \text{if } (X_i,Y_i)=(0,2)\\
0,& \text{if } (X_i,Y_i)=(1,0)\\
1-\delta,& \text{if } (X_i,Y_i)=(1,1)\\
\delta,& \text{if } (X_i,Y_i)=(1,2)\\
\gamma,& \text{if } (X_i,Y_i)=(2,0)\\
0,& \text{if } (X_i,Y_i)=(2,1)\\
1-\gamma,& \text{if } (X_i,Y_i)=(2,2)\\
\end{cases} ,
\end{multline}
and as before $m_i=\left\vert{\left\{ {j\in N_i} \mid X_j=1\right\}}\right\vert=\left\vert{N_i\cap I(t)}\right\vert$. As can be noticed, the first and the third element in Eq. \eqref{MC2_id} have changed, and for $\theta=0$ the model reduces to the non-vaccinating one.

In this infection-dominant model the marginal probabilities are:
\begin{empheq}[box=\fbox]{align}
p_{R,i}&(t+1) =(1-\gamma)p_{R,i}(t)+\delta p_{I,i}(t)\notag\\
&+ (1-\beta)^{m_i}\theta(1-p_{R,i}(t)-p_{I,i}(t)) ,\label{exact_R_id}\\
p_{I,i}&(t+1) =(1-\delta)p_{I,i}(t)\notag\\
&+(1-(1-\beta)^{m_i})(1-p_{R,i}(t)-p_{I,i}(t)) ,\label{exact_I_id}
\end{empheq}
and
\begin{equation}\label{exact_S_id}
p_{S,i}(t+1) =(1-\beta)^{m_i} (1-\theta)(1-p_{R,i}(t)-p_{I,i}(t))+\gamma p_{R,i}(t) ,
\end{equation}
which is again consistent with the fact that $p_{S,i}(t)+p_{I,i}(t)+p_{R,i}(t)=1$ for all $t$.

The steady state behavior in the presence of immunization is rather different from the non-vaccinating case, in which all the node became susceptible. In this model, once there is no node in the infected state, the Markov chain reduces to a simpler Markov chain, where the nodes are all decoupled. In fact from that time on, each node has an independent transition probability between $S$ and $R$. The stationary distribution of each single node is then $P_S^* = \frac{\gamma}{\gamma+\theta}$ and $P_R^* = \frac{\theta}{\gamma+\theta}$ (Fig. \ref{steady}). In order for this MC to converge, we should have $\gamma\theta \neq 1$. The stationary distribution of each state $X$ is then:
$$\pi_X = \prod_{i=1}^n (\frac{\gamma}{\gamma+\theta})^{\mathbb{I}(X_i=0)} \cdot 0^{\mathbb{I}(X_i=1)} \cdot (\frac{\theta}{\gamma+\theta})^{\mathbb{I}(X_i=2)}$$
\begin{figure}[thpb]
  \centering
    \includegraphics[width=0.4\columnwidth]{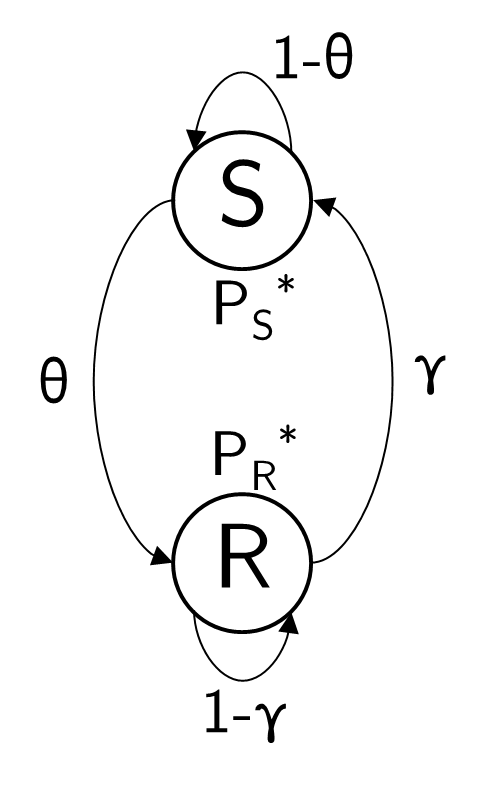}
      \caption{Reduced Markov chain of a single node in the steady state.}
      \label{steady}
\end{figure}

Now the nonlinear map (mean-field approximation of the Markov chain model) can is obtained as:
\begin{empheq}[box=\fbox]{align}
P&_{R,i}(t+1) =(1-\gamma)P_{R,i}(t)+\delta P_{I,i}(t)\notag\\
&+ \prod_{j\in N_i} (1-\beta P_{I,j}(t))\theta(1-P_{R,i}(t)-P_{I,i}(t)), \label{nonlinear_R_id}\\
P&_{I,i}(t+1) =(1-\delta)P_{I,i}(t)+\notag\\
&(1-\prod_{j\in N_i} (1-\beta P_{I,j}(t)))(1-P_{R,i}(t)-P_{I,i}(t)), \label{nonlinear_I_id}
\end{empheq}
and
\begin{multline}\label{nonlinear_S_id}
P_{S,i}(t+1)=\gamma P_{R,i}(t) +\\
\prod_{j\in N_i} (1-\beta P_{I,j}(t)) (1-\theta) (1-P_{R,i}(t)-P_{I,i}(t)) .
\end{multline}

It can be easily verified that one fixed point of this nonlinear map occurs at $P_{R,i}(t)=P_R^*$ and $P_{I,i}(t)=0$, i.e.
$$\begin{bmatrix}P_{R}(t)\\ P_{I}(t)\end{bmatrix}= \begin{bmatrix} \frac{\theta}{\gamma+\theta}1_n\\ 0_n\end{bmatrix}, $$
which is nicely consistent with the steady state of the Markov chain.

By expanding the above model around its fixed point, up to the first order, and doing some algebra, the linear model can be obtained as:
\begin{align}
\tilde{P}_{R,i}(t+1) = & P_R^* + (1-\gamma-\theta)(\tilde{P}_{R,i}(t)-P_R^*)\notag\\
&+(\delta-\theta)\tilde{P}_{I,i}(t) -\theta P_S^*\beta\sum\limits_{j\in N_i} \tilde{P}_{I,j} ,\\
\tilde{P}_{I,i}(t+1) = & (1-\delta)\tilde{P}_{I,i}(t)+ \beta\sum\limits_{j\in N_i} \tilde{P}_{I,j} (1-P_R^*) ,
\end{align}
and the matrix form of:
\begin{empheq}[box=\fbox]{align}
&\begin{bmatrix}\tilde{P}_R(t+1)\\\tilde{P}_I(t+1)\end{bmatrix}=\begin{bmatrix}  P_R^* 1_n\\ 0_n\end{bmatrix} + M' \begin{bmatrix}\tilde{P}_R(t)-P_R^* 1_n\\\tilde{P}_I(t) - 0_n\end{bmatrix} ,\\
&\text{where}\notag\\
&M'=\begin{bmatrix}
(1-\gamma-\theta)I_n & (\delta-\theta) I_n - \theta  P_S^* \beta A\\
0_{n\times n} & (1-\delta)I_n+ P_S^*\beta A
\end{bmatrix} .
\end{empheq}

\subsubsection{Stability of the Fixed Point}
The following result summarizes the stability of the (disease-free) fixed point.
\begin{proposition}
The main fixed point of the nonlinear map (\ref{nonlinear_R_id}, \ref{nonlinear_I_id}) is
\begin{enumerate}[label=\alph*)]
\item locally stable, if $\frac{\gamma}{\gamma+\theta}\frac{\beta}{\delta}\lambda_{max}(A)<1$, and
\item globally stable, if $\frac{\beta}{\delta}\lambda_{max}(A)<1$ .
\end{enumerate}
\end{proposition}
\begin{proof}
$M'$ is in fact the Jacobian matrix of the nonlinear map, and its largest eigenvalue is less than $1$ if the largest eigenvalue of $(1-\delta)I_n+P_S^*\beta A)$ is less than $1$. It follows that the fixed point is locally stable under this condition, and the statement a is true.

Eq. \eqref{nonlinear_I_id} can be upperbounded as:
\begin{align}
P&_{I,i}(t+1) =(1-\delta)P_{I,i}(t)+\notag\\
&(1-\prod_{j\in N_i} (1-\beta P_{I,j}(t)))(1-P_{R,i}(t)-P_{I,i}(t))\notag\\
&\leq (1-\delta)P_{I,i}(t)+ (\beta\sum\limits_{j\in N_i}P_{I,j})(1-P_{R,i}(t)-P_{I,i}(t))\label{upperbound_1_id}\\
&\leq (1-\delta)P_{I,i}(t)+ \beta\sum\limits_{j\in N_i}P_{I,j}\label{upperbound_2_id} ,
\end{align}
which implies the statement b.
\end{proof}
Note that from \eqref{upperbound_1_id} to \eqref{upperbound_2_id} it is not possible to show an upperbound of $(1-\delta)P_{I,i}(t)+ (\beta\sum\limits_{j\in N_i}P_{I,j})(1-P_R^*)$ instead; as it requires $P_{R,i}(t)+P_{I,i}(t)\geq P_R^*$, which is equivalent to $P_{S,i}(t)\leq P_S^*$, that is not true in general. The authors of \cite{prakash2012threshold} have shown the same condition for the local stability, but they do not provide any result on the global stability.
 
\subsubsection{Mixing Time of MC}
We show that the mixing time of the infection-dominant Markov chain is also $O(\log n )$, when $\frac{\beta}{\delta}\lambda_{max}(A)<1$. Vectors $\mu(t)$, $p(t)$ and the matrix $B$ are defined as before.

\begin{theorem}\label{thm_mixing_id}
If $\frac{\beta\lambda_{\max}(A)}{\delta}<1$, the mixing time of the Markov chain whose transition matrix $S$ is described by Eqs. \eqref{MC1_id} and \eqref{MC2_id} is $O(\log n)$.
\end{theorem}
\begin{proof}
First we show that for each $i\in\left\{n+1,n+2,\dots,2n\right\}$, we have $p_i(t+1) \leq (1-\delta)p_i(t) + \beta\sum\limits_{j\in N_i}p_j(t)$.
$f_i\in \mathbb{R}^{2n+1}$ represents the $i^{th}$ unit column vector, as before.
For the sake of convenience, let us drop the time index $(t)$ again.
\begin{align}
\max_{\mu B=p^T, \mu\succeq 0} p_i(t+1) &= \max_{\mu B=p^T, \mu\succeq 0} \mu SBf_i\\
 &= \max_{\mu\succeq 0} \min_\lambda \mu SBf_i-(\mu B-p^T)\lambda\\ \label{minmax_id}
 &= \min_\lambda \max_{\mu\succeq 0} \mu(SBf_i-B\lambda) +p^T\lambda ,
\end{align}
where $\lambda\in\mathbb{R}^{2n+1}$ is the column vector of Lagrange multipliers. By the same argument as in the proof of Theorem \ref{thm_mixing}
\begin{equation}\label{leq_id}
SBf_i-B\lambda\preceq 0 ,
\end{equation}
and for the elements of $SBf_i$ and $B\lambda$ we have:
\begin{align}
&(SBf_i)_X = (SB)_{X,i} = \sum_{Y\in\left\{0,1,2\right\}^n} S_{X,Y}B_{Y,i}\\
 &= \begin{cases} \mathbb{P}\left\{ Y_i=2 \mid X\right\},& i\in\left\{1,2,\dots,n\right\}\\ \mathbb{P}\left\{ Y_{i-n}=1 \mid X\right\},& i\in\left\{n+1,n+2,\dots,2n\right\}\end{cases}\\
 &= \begin{cases}
(1-\beta)^{m_i}\theta,& \text{if } i\in\left\{1,2,\dots,n\right\} \text{ and } X_i=0\\
\delta,& \text{if } i\in\left\{1,2,\dots,n\right\} \text{ and } X_i=1\\
1-\gamma,& \text{if } i\in\left\{1,2,\dots,n\right\} \text{ and } X_i=2\\
1-(1-\beta)^{m_{i-n}},& \text{if } i\in\left\{n+1,\dots,2n\right\} \text{ and } X_{i-n}=0\\
1-\delta,& \text{if } i\in\left\{n+1,\dots,2n\right\} \text{ and } X_{i-n}=1\\ \label{SBf_i_id}
0,& \text{if } i\in\left\{n+1,\dots,2n\right\} \text{ and } X_{i-n}=2\\
\end{cases}\\ \label{Bl_id}
&(B\lambda)_X =\lambda_0 + \sum\limits_{k=1}^n B_{X,k}\lambda_k + \sum\limits_{k=n+1}^{2n} B_{X,k}\lambda_k .
\end{align}
As mentioned before, we are interested to evaluate $p_i(t+1)$ only for $i\in\left\{n+1,n+2,\dots,2n\right\}$. Since the corresponding terms in \eqref{SBf_i_id} (the lower three) do not depend on $\theta$, the equations for optimal Lagrange multipliers are the same as in Theorem \ref{thm_mixing}, which lead to
\begin{equation}
p_{I,i}(t+1) \leq (1-\delta)p_{I,i}(t) + \beta\sum\limits_{j\in N_i}p_{I,j}(t) .
\end{equation}
and consequently
\begin{equation}
p_I(t+1)\preceq ((1-\delta)I_n+\beta A) p_I(t) .
\end{equation}

Now for any $t< t_{mix}(\epsilon)$:
\begin{align}
\epsilon &< \mathbb{P}\left\{ \substack{\text{some nodes are infected at time $t$} \mid \\\text{all nodes were infected at time $0$}} \right\}\\
 &\leq \sum\limits_{i=1}^n p_{I,i}(t) = 1_{n}^T p_I(t)\\
 &\leq 1_{n}^T ((1-\delta)I_n+\beta A)^t p_I(0)\\
 &\leq 1_{n}^T ((1-\delta)I_n+\beta A)^t 1_{n}\\
 &\leq \|1_{n}\|^2 \|(1-\delta)I_n+\beta A\|^t\\
 &=n\|(1-\delta)I_n+\beta A\|^t .
\end{align}
$\|(1-\delta)I_n+\beta A\|<1$ leads to the fact that $t< \frac{\log \frac{n}{\epsilon}}{-\log \|(1-\delta)I_n+\beta A\|}$ for all $t< t_{mix}(\epsilon)$. Therefore $t_{mix}(\epsilon)\leq \frac{\log \frac{n}{\epsilon}}{-\log \|(1-\delta)I_n+\beta A\|}$, which means the mixing time is $O(\log n)$.
\end{proof}

\subsection{Vaccination-Dominant Model}

In this variation of the model the assumption is if a susceptible node receives both infection and vaccine at the same time, it becomes vaccinated. The transition probabilities of the Markov chain are again
\begin{align}\label{MC1_vd}
S_{X,Y}&= \mathbb{P}\left\{\xi(t+1)=Y \mid \xi(t)= X\right\}\notag\\
&= \prod_{i=1}^n \mathbb{P}\left\{\xi_i(t+1)=Y_i \mid \xi(t) = X\right\} ,
\end{align}
with the change that
\begin{multline}\label{MC2_vd}
\mathbb{P}\left\{\xi_i(t+1)=Y_i \mid \xi(t) = X\right\}=\\
\begin{cases}
(1-\beta)^{m_i}(1-\theta),& \text{if } (X_i,Y_i)=(0,0)\\
(1-(1-\beta)^{m_i})(1-\theta),& \text{if } (X_i,Y_i)=(0,1)\\
\theta,& \text{if } (X_i,Y_i)=(0,2)\\
0,& \text{if } (X_i,Y_i)=(1,0)\\
1-\delta,& \text{if } (X_i,Y_i)=(1,1)\\
\delta,& \text{if } (X_i,Y_i)=(1,2)\\
\gamma,& \text{if } (X_i,Y_i)=(2,0)\\
0,& \text{if } (X_i,Y_i)=(2,1)\\
1-\gamma,& \text{if } (X_i,Y_i)=(2,2)\\
\end{cases} ,
\end{multline}
and $m_i=\left\vert{\left\{ {j\in N_i} \mid X_j=1\right\}}\right\vert=\left\vert{N_i\cap I(t)}\right\vert$ as before.

In this case the marginal probabilities are:
\begin{empheq}[box=\fbox]{align}
p_{R,i}&(t+1) =(1-\gamma)p_{R,i}(t)+\delta p_{I,i}(t)+\notag\\
& \theta(1-p_{R,i}(t)-p_{I,i}(t)) ,\label{exact_R_vd}\\
p_{I,i}&(t+1) =(1-\delta)p_{I,i}(t)+\notag\\
&(1-\theta)(1-(1-\beta)^{m_i})(1-p_{R,i}(t)-p_{I,i}(t))\label{exact_I_vd}
\end{empheq}

The nonlinear map, or the mean-field approximation, can be stated as:
\begin{empheq}[box=\fbox]{align}
&P_{R,i}(t+1) =(1-\gamma)P_{R,i}(t)+\delta P_{I,i}(t)\notag\\
&+ \theta(1-P_{R,i}(t)-P_{I,i}(t)), \label{nonlinear_R_vd}\\
&P_{I,i}(t+1) =(1-\delta)P_{I,i}(t)+(1-\theta)\notag\\
&\cdot(1-\prod_{j\in N_i} (1-\beta P_{I,j}(t)))(1-P_{R,i}(t)-P_{I,i}(t)) \label{nonlinear_I_vd}
\end{empheq}

As a result, the first order (linear) model is:
\begin{align}
\tilde{P}_{R,i}(t+1) = & P_R^* + (1-\gamma-\theta)(\tilde{P}_{R,i}(t)-P_R^*)\notag\\
&+(\delta-\theta)\tilde{P}_{I,i}(t) -\theta P_S^*\beta\sum\limits_{j\in N_i} \tilde{P}_{I,j} ,\\
\tilde{P}_{I,i}(t+1) = & (1-\delta)\tilde{P}_{I,i}(t)+(1-\theta) P_S^*\beta\sum\limits_{j\in N_i} \tilde{P}_{I,j} ,
\end{align}
or the following matrix form:
\begin{empheq}[box=\fbox]{align}
&\begin{bmatrix}\tilde{P}_R(t+1)\\\tilde{P}_I(t+1)\end{bmatrix}=\begin{bmatrix} P_R^* 1_n\\ 0_n\end{bmatrix} + M'' \begin{bmatrix}\tilde{P}_R(t)-P_R^* 1_n\\\tilde{P}_I(t) - 0_n\end{bmatrix} ,\\
&\text{where}\notag\\
&M''=\begin{bmatrix}
(1-\gamma-\theta)I_n & (\delta-\theta) I_n - \theta  P_S^* \beta A\\
0_{n\times n} & (1-\delta)I_n+(1-\theta) P_S^*\beta A
\end{bmatrix} .
\end{empheq}
We should note that for the vaccination-dominant model, the steady state of the Markov chain and the main fixed point of the mapping are exactly the same as in  the infection-dominant model. However, as we may expect, the vaccination-dominant model is more stable.

\subsubsection{Stability of The Fixed Point}
The stability of the vaccination-dominant model can be summarized in the following theorem.
\begin{proposition}
The main fixed point of the nonlinear map (\ref{nonlinear_R_vd}, \ref{nonlinear_I_vd}) is
\begin{enumerate}[label=\alph*)]
\item locally stable, if $(1-\theta)\frac{\gamma}{\gamma+\theta}\frac{\beta}{\delta}\lambda_{max}(A)<1$, and
\item globally stable, if $(1-\theta)\frac{\beta}{\delta}\lambda_{max}(A)<1$ .
\end{enumerate}
\end{proposition}
\begin{proof}
The statement a is again clear since if the largest eigenvalue of $(1-\delta)I_n+(1-\theta) P_S^*\beta A$ is less than one, then the largest eigenvalue of $M''$ is less than $1$, which means the norm of the Jacobian matrix is less than $1$.

The statement b also follows from upperbounding Eq. \eqref{nonlinear_I_vd} as
\begin{equation}
P_{I,i}(t+1) \leq (1-\delta)P_{I,i}(t)+ (1-\theta)\beta\sum\limits_{j\in N_i}P_{I,j}\label{upperbound_2_vd} .
\end{equation}
\end{proof}

\begin{figure*}[!t]
\centering
\subfloat{\includegraphics[width=3in]{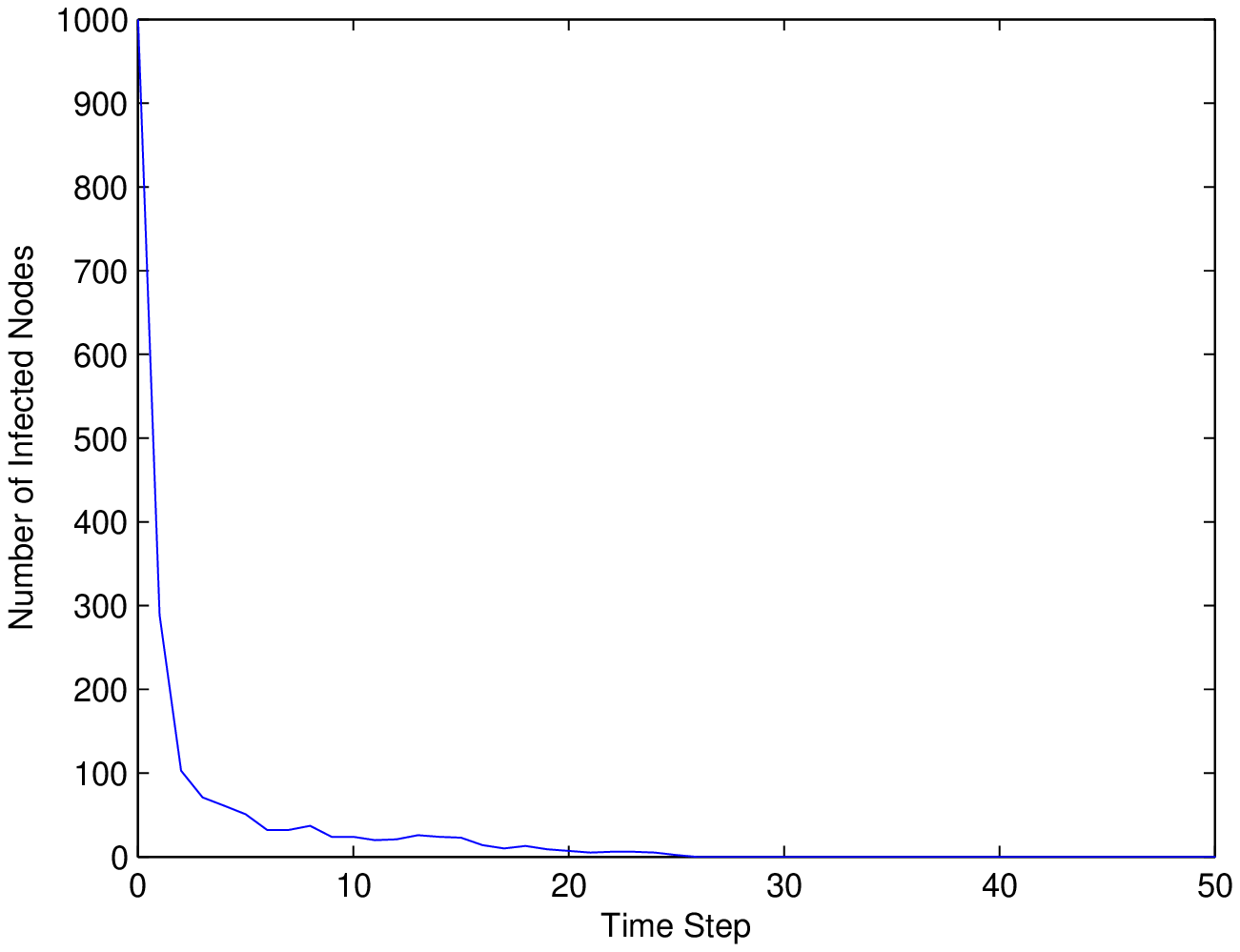}%
\label{plot1}}
\subfloat{\includegraphics[width=3in]{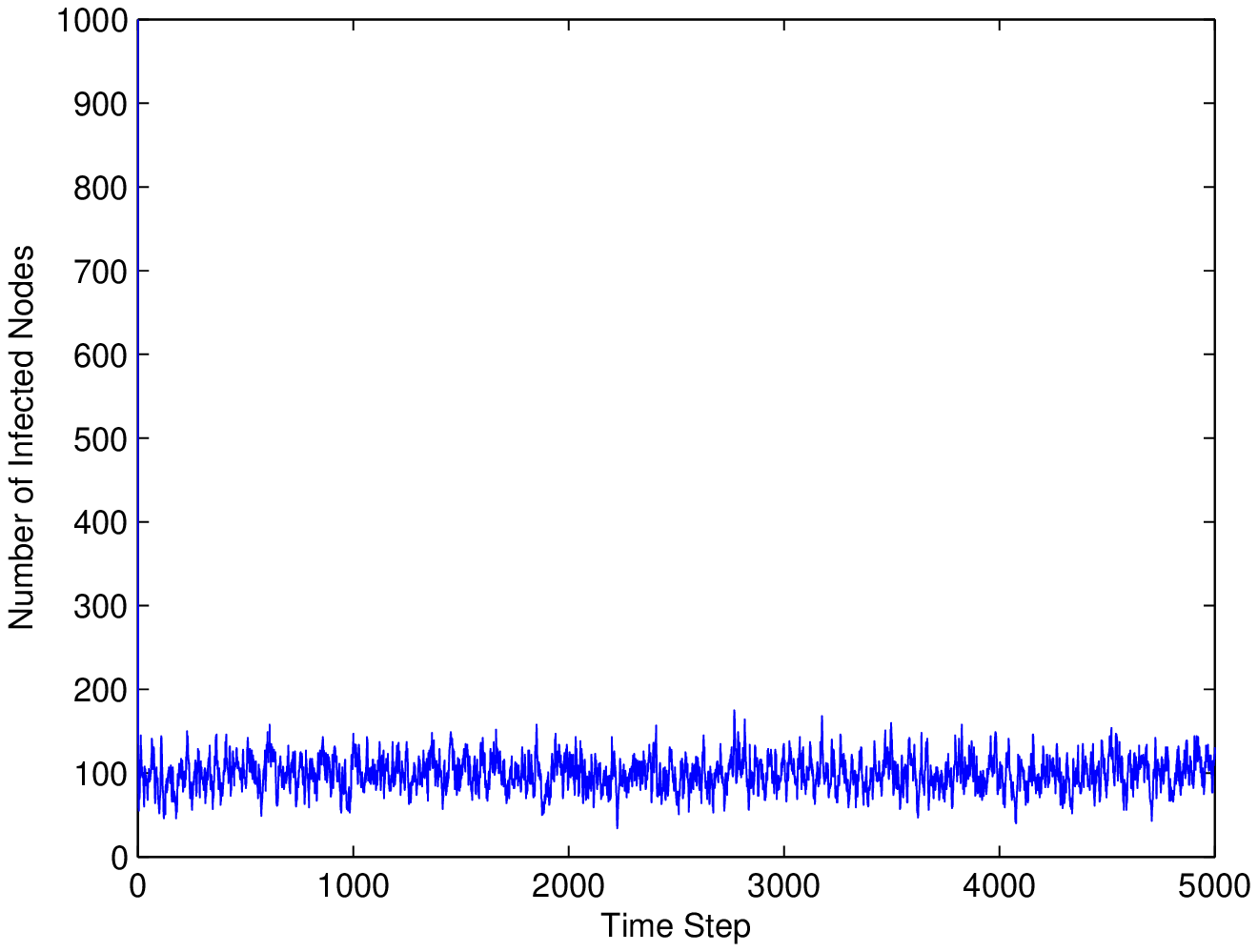}%
\label{plot2}}
\caption{The evolution of SIRS/SIV epidemics over an Erd\H{o}s-R\'enyi graph with $n=2000$ nodes and the same initial states. a) $\frac{\gamma}{\gamma+\theta}\frac{\beta\|A\|}{\delta}<1$: Fast eradication of the epidemic. b) $\frac{\gamma}{\gamma+\theta}\frac{\beta\|A\|}{\delta}>1$: Epidemic spread around the nontrivial fixed point (convergence is not observed.)}
\label{joint}
\end{figure*}

\subsubsection{Mixing Time of MC}
As shown above, the stability condition of the main fixed point (epidemic eradication) is relaxed by a factor of $(1-\theta)$ in the vaccination-dominant model. In this part, we show that the condition for the fast mixing time of the Markov chain is also relieved by the same factor.
\begin{theorem}
If $(1-\theta)\frac{\beta\lambda_{\max}(A)}{\delta}<1$, the mixing time of the Markov chain whose transition matrix $S$ is described by Eqs. \eqref{MC1_vd} and \eqref{MC2_vd} is $O(\log n)$.
\end{theorem}
\begin{proof}
We use the same linear programming argument as in the proofs of Theorems \ref{thm_mixing} and \ref{thm_mixing_id}, and show that for each $i\in\left\{n+1,n+2,\dots,2n\right\}$, we have $p_i(t+1) \leq (1-\delta)p_i(t) + (1-\theta)\beta\sum\limits_{j\in N_i}p_j(t)$.
\begin{align}
&(SBf_i)_X = \begin{cases} \mathbb{P}\left\{ Y_i=2 \mid X\right\},& i\in\left\{1,2,\dots,n\right\}\\ \mathbb{P}\left\{ Y_{i-n}=1 \mid X\right\},& i\in\left\{n+1,\dots,2n\right\}\end{cases}=\notag\\
 & \begin{cases}
\theta,& \text{if } i\in\left\{1,2,\dots,n\right\} \text{ and } X_i=0\\
\delta,& \text{if } i\in\left\{1,2,\dots,n\right\} \text{ and } X_i=1\\
1-\gamma,& \text{if } i\in\left\{1,2,\dots,n\right\} \text{ and } X_i=2\\
\!\begin{aligned}(1-\theta)(1-(1\ \ \\-\beta)^{m_{i-n}}),\end{aligned}& \text{if } i\in\left\{n+1,\dots,2n\right\} \text{ and } X_{i-n}=0\\
1-\delta,& \text{if } i\in\left\{n+1,\dots,2n\right\} \text{ and } X_{i-n}=1\\ \label{SBf_i_vd}
0,& \text{if } i\in\left\{n+1,\dots,2n\right\} \text{ and } X_{i-n}=2\\
\end{cases}\\ \label{Bl_vd}
&(B\lambda)_X =\lambda_0 + \sum\limits_{k=1}^n B_{X,k}\lambda_k + \sum\limits_{k=n+1}^{2n} B_{X,k}\lambda_k .
\end{align}
We claim that the Lagrange multiplier vector $\lambda^*=[\lambda_0^*,\lambda_1^*,\dots,\lambda_{2n}^*]^T$ with the following values is in the feasible set:
\begin{equation}
\begin{cases}
\lambda_0^*=0\\
\lambda_1^*=\dots=\lambda_n^*=0\\
\lambda_{n+i}^*=1-\delta\\
\lambda_{n+j}=\beta(1-\theta) \text{  for } j\in N_i\\
\lambda_{n+j}=0 \text{  for } j\not\in N_i
\end{cases}
\end{equation}
Verification of the claim for all possible cases is as follows.\\
For $X_i=0, \left\vert{N_i\cap I(t)}\right\vert=m$ :
\begin{multline}
\mathbb{P}\left\{ Y_i=1 \mid X\right\}=(1-\theta)(1-(1-\beta)^m) \leq\\
m\beta(1-\theta) = \lambda_0^* + \sum\limits_{k=1}^n B_{X,k}\lambda_k^* + \sum\limits_{k=n+1}^{2n} B_{X,k}\lambda_k^* .
\end{multline}
For $X_i=1, \left\vert{N_i\cap I(t)}\right\vert=m$ :
\begin{multline}
\mathbb{P}\left\{ Y_i=1 \mid X\right\}=1-\delta \leq 1-\delta+m\beta(1-\theta) \\
= \lambda_0^* + \sum\limits_{k=1}^n B_{X,k}\lambda_k^* + \sum\limits_{k=n+1}^{2n} B_{X,k}\lambda_k^* .
\end{multline}
For $X_i=2, \left\vert{N_i\cap I(t)}\right\vert=m$ :
\begin{multline}
\mathbb{P}\left\{ Y_i=1 \mid X\right\}=0 \leq\\
m\beta(1-\theta) = \lambda_0^* + \sum\limits_{k=1}^n B_{X,k}\lambda_k^* + \sum\limits_{k=n+1}^{2n} B_{X,k}\lambda_k^* .
\end{multline}
It follows that
\begin{align}
\max_{\mu B=p^T, \mu\succeq 0} &p_i(t+1)= \min_\lambda \max_{\mu\succeq 0} \mu(SBf_i-B\lambda) +p^T\lambda\notag\\
 &\leq p^T\lambda^* = (1-\delta)p_i +\beta(1-\theta)\sum\limits_{j\in N_i}p_j ,
\end{align}
which proves
\begin{equation}
p_I(t+1)\preceq ((1-\delta)I_n+\beta(1-\theta) A) p_I(t) .
\end{equation}
Under the condition that $\frac{\beta(1-\theta)\lambda_{\max}(A)}{\delta}<1$, by the same argument as in the proof of Theorem \ref{thm_mixing_id}, $t_{mix}(\epsilon)\leq \frac{\log \frac{n}{\epsilon}}{-\log \|(1-\delta)I_n+\beta(1-\theta) A\|}$.
\end{proof}


\section{Experimental Results}

We show the simulation results on Erd\H{o}s-R\'enyi graphs, for the epidemic thresholds below and above $1$, and they confirm the theorems proved in the paper. As it can be seen in Fig. \ref{plot1}, for SIRS epidemics ($\theta=0$), when the condition $\frac{\beta\|A\|}{\delta}<1$ is satisfied the epidemic decays exponentially, and dies out quickly. In contrast when $\frac{\beta\|A\|}{\delta}>1$, the epidemic does not exhibit convergence to the disease-free state in any observable time. Fig. \ref{plot2} illustrates this phenomenon, and indicates that the epidemic keeps spreading around its nontrivial fixed point.

For the first SIV model (infection-dominant), we observe the same exponential decay (Fig. \ref{plot1}), when $\frac{\gamma}{\gamma+\theta}\frac{\beta\|A\|}{\delta}<1$, which means the vaccination indeed makes the system more stable. Furthermore, for the vaccination-dominant model, under $(1-\theta)\frac{\gamma}{\gamma+\theta}\frac{\beta\|A\|}{\delta}<1$, we observe the fast convergence again, which confirms that the system is even more stable in this case. As might have been speculated, for $\frac{\gamma}{\gamma+\theta}\frac{\beta}{\delta}\lambda_{max}(A)>1$ in the infection-dominant and $(1-\theta)\frac{\gamma}{\gamma+\theta}\frac{\beta}{\delta}\lambda_{max}(A)>1$ in the vaccination-dominant, we are not able to see epidemic eradication in any reasonable time, and we obtain similar plots as Fig. \ref{plot2}.

\section{Conclusions}

We studied the exact network-based Markov chain Model for the SIRS/SIV epidemics, and their celebrated mean-field approximation. We showed that the threshold conditions coincide for fast-mixing of the exact Markov chain and the stability of the mean-field approximation at the disease-free fixed point. Furthermore, we showed for above-threshold epidemics, that there exists a unique nontrivial fixed point corresponding to the endemic state. Interestingly, the simulations suggest that in the latter case, the underlying Markov chain model should also have an exponentially slow mixing time, which leads to the conjecture that the threshold condition is indeed tight.



\bibliographystyle{IEEEtran}
\bibliography{references}

\end{document}